\documentclass[envcountsame]{llncs}
\usepackage{ifthen}
\newcommand{\techRep}{true} 
\newcommand{\iftechrep}{\ifthenelse{\equal{\techRep}{true}}}

\usepackage{amsfonts,amstext,amsmath,amssymb,stmaryrd,mathrsfs,calc,xspace,mathtools}
\usepackage{hyperref}
\usepackage{complexity}

\usepackage{algorithm}
\usepackage[noend]{algorithmic}

\usepackage{enumerate}
\usepackage{macros}
\usepackage{mathtools}
\usepackage{datetime}
\usepackage{colortbl}
\usepackage{hhline}
\usepackage{float}

\usepackage{xcolor,pgf,tikz}
\usetikzlibrary{arrows,automata,shapes,snakes,calc,positioning}
\usepackage{pgfplots}

\makeatletter
\newenvironment{customlegend}[1][]{%
    \begingroup
    \pgfplots@init@cleared@structures
    \pgfplotsset{#1}%
}{%
    \pgfplots@createlegend
    \endgroup
}%
\def\addlegendimage{\pgfplots@addlegendimage}
\makeatother

\usepackage{etoolbox}
\begin{document}

\title{Approaching the Coverability Problem Continuously}
\author{Michael Blondin\inst{1,2}\thanks{Supported by the Fonds
    qu\'{e}b\'{e}cois de la recherche sur la nature et les
    technologies (FQRNT), by the French Centre national de la
    recherche scientifique (CNRS), and by the ``Chaire Digiteo, ENS
    Cachan --- {\'E}cole Polytechnique''.} \and Alain Finkel\inst{2}
  \and Christoph Haase\inst{2}\thanks{Supported by Labex Digicosme,
    Univ. Paris-Saclay, project VERICONISS.} \and Serge
  Haddad\inst{2,3}\thanks{Supported by ERC project EQualIS (FP7-308087).}} \institute{DIRO, Universit\'{e} de Montr\'{e}al,
  Canada \and LSV, CNRS \& ENS Cachan, Universit\'e Paris-Saclay,
  France \and Inria, France}

\maketitle
\begin{abstract}
  The coverability problem for Petri nets plays a central role in the
  verification of concurrent shared-memory programs. However, its high
  EXPSPACE-complete complexity poses a challenge when encountered in
  real-world instances. In this paper, we develop a new approach to
  this problem which is primarily based on applying forward
  coverability in continuous Petri nets as a pruning criterion inside
  a backward-coverability framework. A cornerstone of our approach is
  the efficient encoding of a recently developed polynomial-time
  algorithm for reachability in continuous Petri nets into SMT. We
  demonstrate the effectiveness of our approach on standard benchmarks
  from the literature, which shows that our approach decides
  significantly more instances than any existing tool and is in
  addition often much faster, in particular on large instances.
\end{abstract}

\section{Introduction}

Counter machines and Petri nets are popular mathematical models for
modeling and reasoning about distributed and concurrent systems. They
provide a high level of abstraction that allows for employing them in
a great variety of application domains, ranging, for instance, from
modeling of biological, chemical and business processes to the formal
verification of concurrent programs.

Many safety properties of real-world systems reduce to the
\emph{coverability problem} in Petri nets: Given an initial and a
target configuration, does there exist a sequence of transitions
leading from the initial configuration to a configuration larger than
the target configuration?
%
%
For instance, in an approach pioneered by German and
Sistla~\cite{GS92} multi-threaded non-recursive finite-state programs
with shared variables, which naturally occur in
predicate-abstraction-based verification frameworks, are modeled as
Petri nets such that every program location corresponds to a place in
a Petri net, and the number of tokens of a place indicates how many
threads are currently at the corresponding program
location. Coverability can then, for instance, be used in order to
detect whether a mutual exclusion property could be violated when a
potentially unbounded number of threads is executed in parallel.
The coverability problem was one of the first decision problems for
Petri nets that was shown decidable
and \EXPSPACE-complete~\cite{KM67,CLM76,Rack78}. Despite this huge
worst-case complexity, over the course of the last twenty years, a
plethora of tools has emerged that have shown to be able to cope with
a large number of real-world instances of coverability problems in a
satisfactory manner.

\subsubsection{Our contribution.}
We present a new approach to the coverability problem and its
implementation. When run on standard benchmarks that we obtained from
the literature, our approach proves more than 91\% of safe instances
to be safe, most of the time much faster when compared to existing
tools, and none of those tools can individually prove more than 84\%
of safe instances to be safe. We additionally demonstrate that our
approach is also competitive when run on unsafe instances. In
particular, it decides 142 out of 176 (80\%) instances of our
benchmark suite, while the best competitor only decides 122 (69\%)
instances.

Our approach is conceptually extremely simple and exploits recent
advances in the theory of Petri nets as well as the power of modern
SMT-solvers inside a backward-coverability framework. In~\cite{FH15},
Fraca and Haddad solved long-standing open problems about the
complexity of decision problems for so-called continuous Petri nets.
This class was introduced by David and Alla~\cite{DA87} and allows for
transitions to be fired a non-negative real number of times---hence
places may contain a non-negative real number of tokens. The
contribution of~\cite{FH15} was to present polynomial-time algorithms
that decide all of coverability, reachability and boundedness in this
class. A further benefit of~\cite{FH15} is to show that continuous
Petri nets over the reals are equivalent to continuous Petri nets over
the rationals, and, moreover, to establish a set of simple sufficient
and necessary conditions in order to decide reachability in continuous
Petri nets. The first contribution of our paper is to show that these
conditions can efficiently be encoded into a sentence of \emph{linear
  size} in the existential theory of the non-negative rational numbers
with addition and order (FO($\Qpos, +, >$)). This encoding paves the
way for deciding coverability in continuous Petri nets inside
SMT-solvers and is particularly useful in order to efficiently answer
\emph{multiple coverability queries} on the same continuous Petri net
due to caching strategies present in modern SMT-solvers. Moreover, we
show that our encoding in effect \emph{strictly subsumes} a recently
introduced CEGAR-based approach to coverability described by Esparza
\emph{et al.} in~\cite{ELMMN14}; in particular we can completely avoid
the potentially exponentially long CEGAR-loop, cf.\ the related work
section below. The benefit of coverability in continuous Petri nets is
that it provides a way to over-approximate coverability under the
standard semantics: any configuration that is not coverable in a
continuous Petri net is also not coverable under the standard
semantics. This observation can be exploited inside a
backward-coverability framework as follows. Starting at the target
configuration to be covered, the classical backward-coverability
algorithm~\cite{ACJT00}
repeatedly computes the set of all minimal predecessor configurations
that by application of one transition cover the target or some earlier
computed configuration until a fixed point is reached, which is
guaranteed to happen due to Petri nets being well-structured
transition systems~\cite{FS01}. The crux to the performance of the
algorithm lies in the size of the set of minimal elements that is
computed during each iteration, which may grow
exponentially.\footnote{This problem is commonly referred to as the
  \emph{symbolic state explosion problem}, cf.~\cite{DRB04}.} This is
where continuous coverability becomes beneficial. In our approach, if
a minimal element is not continuously coverable, it can safely be
discarded since none of its predecessors is going to be coverable
either, which substantially shrinks the predecessor set.
In effect, this heuristic yields a powerful pruning technique,
enabling us to achieve the aforementioned advantages when compared to
other approaches on standard benchmarks. 

\iftechrep{Some proof details are only sketched in the main part of
  this paper; full details can be found in the appendix.}{Due to
  space constraints, we only sketch some of the proofs in this paper.
Full details can be found in~\cite{BFHH15}.}

\subsubsection{Related Work.}
Our approach is primarily related to the work by Esparza \emph{et
  al.}~\cite{ELMMN14}, by Kaiser, Kroening and Wahl~\cite{KKW14}, and
by Delzanno, Raskin and van Begin~\cite{DRB01}. In~\cite{ELMMN14},
Esparza \emph{et al.} presented an implementation of a semi-decision
procedure for disproving coverability which was originally proposed by
Esparza and Melzer~\cite{EM00}. It is based on the Petri-net state
equation and traps as sufficient criteria in order to witness
non-coverability. As shown in~\cite{EM00}, those conditions can be
encoded into an equi-satisfiable system of linear inequalities called
the \emph{trap inequation} in~\cite{EM00}. This approach is, however,
prone to numerical imprecision that become problematic even for
instances of small size~\cite[Sec.~5.3]{EM00}. For that reason, the
authors of~\cite{ELMMN14} resort to a CEGAR-based variant of the
approach described in~\cite{EM00} which has the drawback that in the
worst case, the CEGAR loop has to be executed an exponential number of
times leading to an exponential number of queries to the underlying
SMT-solver. We will show in Section~\ref{ssec:esparza} that the
conditions used in~\cite{ELMMN14} are strictly subsumed by a subset of
the conditions required to witness coverability in continuous Petri
nets: whenever the procedure described in~\cite{ELMMN14} returns
uncoverable then coverability does not hold in the continuous setting
either, but not \emph{vice versa}. Thus, a single satisfiability check
to our formula in existential FO($\Qpos, +, >$) encoding continuous
coverability that we develop in this paper completely subsumes the
CEGAR-approach presented in~\cite{ELMMN14}. Another difference
to~\cite{ELMMN14} is that here we present a sound \emph{and} complete
decision procedure.

Regarding the relationship of our work to~\cite{KKW14}, Kaiser
\emph{et al.} develop in their paper an approach to coverability in
richer classes of well-structured transition systems that is also
based on the classical backward-analysis algorithm. They also employ
forward analysis in order to prune the set of minimal elements during
the backward iteration, and in addition a widening heuristic in order
to over-approximate the minimal basis. Our approach differs in that
our minimal basis is always precise yet as small as possible modulo
continuous coverability. Thus no backtracking as in~\cite{KKW14} is
needed, which is required when the widened basis turns out to be too
inaccurate. Another difference is that for the forward analysis, a
Karp-Miller tree is incrementally built in the approach described
in~\cite{KKW14}, whereas we use the continuous coverability
over-approximation of coverability.

The idea of using an over-approximation of the reachability set of a
Petri net in order to prune minimal basis elements inside a backward
coverability framework was first described by Delzanno \emph{et
  al.}~\cite{DRB01}, who use place invariants as a pruning
criterion. However, computing such invariants and checking if a
minimal basis element can be pruned potentially requires exponential
time.

Finally, a number of further techniques and tools for deciding Petri
net coverability or more general well-structured transition systems
have been described in the literature. They are, for instance, based
on efficient data structures~\cite{Gan02,FRSB02,DRB04,GMDKRB07} and
generic algorithmic frameworks such as EEC~\cite{GRB06} and
IC3~\cite{KMNP13}.




\section{Preliminaries}\label{sec:preliminaries}

We denote by $\Q$, $\Z$ and $\N$ the set of rationals, integers, and
natural numbers, respectively, and by $\Qpos$ the set of non-negative
rationals. Throughout the whole paper, numbers are encoded in binary,
and rational numbers as pairs of integers encoded in binary. Let $\D
\subseteq \Q$, $\D^E$ denotes the set of vectors indexed by a finite
set $E$.  A vector $\vec{u}$ is denoted by $\vec{u}=(u_i)_{i\in E}$.
Given vectors $\vec{u} = (u_i)_{i\in E}, \vec{v} = (v_i)_{i\in E}\in
\D^E$, addition $\vec{u} + \vec{v}$ is defined component-wise, and
$\vec{u} \le \vec{v}$ whenever $u_i \le v_i$ for all $i\in
E$. Moreover, $\vec{u} < \vec{v}$ whenever $\vec{u} \le \vec{v}$ and
$\vec{u} \neq \vec{v}$.  Let $E' \subseteq E$ and $\vec{v}\in \D^{E}$,
we sometimes write $\vec{v}[E']$ as an abbreviation for $(v_i)_{i \in
  E'}$. The \emph{support of $\vec{v}$} is the set $\eval{\vec{v}}
\defeq \{ i \in E : \vec{v}_i \neq 0 \}$.

Given finite sets of indices $E$ and $F$, and $\D \subseteq \Q$,
$\D^{E\times F}$ denotes the set of matrices over $\D$ with rows and
columns indexed by elements from $E$ and $F$, respectively. Let
$\mat{M}\in \D^{E\times F}$, $E'\subseteq E$ and $F' \subseteq F$, we
denote by $\mat{M}_{E'\times F'}$ the $\D^{E'\times F'}$ sub-matrix
obtained from $\mat{M}$ whose row and columns indices are restricted
respectively to $E'$ and $F'$. 

\subsubsection{Petri Nets.}

In what follows, we introduce the syntax and semantics of Petri
nets. While we provide a single syntax for nets, we introduce a
discrete (\ie{} in $\N$) and a continuous (\ie{} in $\Qpos$) semantics.

\begin{definition}
  A Petri net is a tuple $\net = (P, T, \Pre, \Post)$, where $P$ is a
  finite set of \emph{places}; $T$ is a finite set of
  \emph{transitions} with $P\cap T = \emptyset$; and $\Pre, \Post \in
  \N^{P\times T}$ are the \emph{backward} and \emph{forward incidence
    matrices}, respectively.
\end{definition}
A (discrete) \emph{marking} of $\net$ is a vector of $\N^P$. A
\emph{Petri net system (PNS)} is a pair $\pns = (\net, \vec{m}_0)$,
where $\net$ is a Petri net and $\vec{m}_0 \in \N^P$ is the
\emph{initial marking}.  The \emph{incidence matrix} $\mat{C}$ of
$\net$ is the $P \times T$ integer matrix defined by $\mat{C} \defeq
\Post - \Pre$. The \emph{reverse net} of $\net$ is $\net^{-1} \defeq
(P, T, \Post, \Pre)$. Let $p \in P$ and $t \in T$, the \emph{pre-sets}
of $p$ and $t$ are the sets $\pre{p} \defeq \{t' \in T : \Post(p,t') >
0\}$ and $\pre{t}\defeq \{ p' \in P : \Pre(p',t) > 0\}$,
respectively. Likewise, the \emph{post-sets} of $p$ and $t$ are
$\post{p} \defeq \{t' \in T : \Pre(p, t') > 0\}$ and $\post{t} = \{p'
\in P : \Post(p', t) > 0\}$, respectively. Those definitions can
canonically be lifted to subsets of places and of transitions, \eg,
for $Q\subseteq P$ we have $\pre{Q} = \bigcup_{p\in Q} \pre{p}$.  We
also introduce the \emph{neighbors} of a subset of places/transitions
by: $\prepost{Q}=\pre{Q}\cup \post{Q}$. Let $S\subseteq T$, then
$\net_{S}$ is the sub-net defined by $\net_{S} \defeq (\prepost{S}, S,
\Pre_{\prepost{S}\times S}, \Post_{\prepost{S}\times S})$.

We say that a transition $t \in T$ is \emph{enabled} at a marking
$\vec{m}$ whenever $\vec{m}(p) \ge \Pre(p, t)$ for every $p \in
\pre{t}$. A transition $t$ that is enabled can be \emph{fired},
leading to a new marking $\vec{m}'$ such that for all places $p \in
P$, $\vec{m}'(p) = \vec{m}(p) + \mat{C}(p, t)$. We write $\vec{m}
\xrightarrow{t} \vec{m}'$ whenever $t$ is enabled at $\vec{m}$ leading
to $\vec{m}'$, and write $\vec{m} \xrightarrow{} \vec{m}'$ if $\vec{m}
\xrightarrow{t} \vec{m}'$ for some $t \in T$. By $\xrightarrow{}^*$ we
denote the reflexive transitive closure of $\xrightarrow{}$. A word
$\sigma = t_1 t_2 \cdots t_k \in T^*$ is a \emph{firing sequence of
  $(\net, \vec{m}_0)$} whenever there exist markings $\vec{m}_1,
\ldots, \vec{m}_k$ such that
\[
\vec{m}_0 \xrightarrow{t_1} \vec{m}_1 \xrightarrow{t_2} \cdots
\xrightarrow{t_{k-1}} \vec{m}_{k-1} \xrightarrow{t_k} \vec{m}_k.
\]
Given a marking $\vec{m}$, the \emph{reachability} problem asks
whether $\vec{m}_0 \xrightarrow{}^* \vec{m}$. The reachability problem
is decidable,
\EXPSPACE-hard~\cite{CLM76} and in ${\bf
  F}_{\omega^3}$~\cite{LerouxS15}, a non-primitive-recursive
complexity class. In this paper, however, we are interested in
deciding coverability, an
\EXPSPACE-complete problem~\cite{CLM76,Rack78}.

\begin{definition}
  Given a Petri net system $\pns = (P, T, \Pre, \Post, \vec{m}_0)$ and
  a marking $\vec{m} \in \N^P$, the \emph{coverability} problem asks
  whether $\vec{m}_0 \xrightarrow{}^* \vec{m}'$ for some $\vec{m}'\ge
  \vec{m}$.
\end{definition}
%

\emph{Continuous Petri nets} are Petri nets in which markings may
consist of rational numbers\footnote{In fact, the original definition
  allows for real numbers, however for studying decidability and
  complexity issues, rational numbers are more convenient.}, and in
which transitions may be fired a fractional number of times. Formally,
a marking of a continuous Petri net is a vector $\vec{m}\in
\Qpos^P$. Let $t \in T$, the \emph{enabling degree} of $t$ with
respect to $\vec{m}$ is a function $\enab{t}{\vec{m}} \in \Qpos \cup
\{\infty \}$ defined by:
\[
\enab{t}{\vec{m}} \defeq 
\begin{cases}
  \min \{\vec{m}(p) / \Pre(p,t) : p \in \pre t\} & 
  \text{if } \pre{t} \neq \emptyset, \\
  \infty & \text{otherwise.}
\end{cases}
\]
We say that $t$ is \emph{$\Q$-enabled} at $\vec{m}$ if
$\enab{t}{\vec{m}} > 0$. If $t$ is $\Q$-enabled it may be \emph{fired}
by any amount $q \in \Qpos$ such that $0 \le q \le \enab{t}{\vec{m}}$,
leading to a new marking $\vec{m}'$ such that for all places $p \in
P$, $\vec{m}(p)' \defeq \vec{m}(p) + q \cdot \mat{C}(p,t)$. In this
case, we write $\vec{m} \xrightarrow{q \cdot t} \vec{m}'$. The
definition of a \emph{$\Q$-firing sequence} $\sigma = q_1t_1 \cdots
q_kt_k \in (\Qpos \times T)^*$ is analogous to the standard definition
of firing sequence, and so are $\xrightarrow{}_\Q$,
$\xrightarrow{}^{*}_{\Q}$ and $\Q$-reachability. The \emph{$\Q$-Parikh
  image} of the firing sequence $\sigma$ is the vector
$\parikh(\sigma)\in \Qpos^T$ such that $\parikh(\sigma)(t) \defeq
\sum_{t_i=t} q_i$.  We also adapt the decision problems for Petri
nets.

\begin{definition}
  Given a Petri net system $\pns = (P, T, \Pre, \Post, \vec{m}_0)$ and
  a marking $\vec{m} \in \Qpos^P$, the $\Q$-reachability (respectively
  $\Q$-coverability) problem asks whether $\vec{m}_0
  \xrightarrow{}_\Q^* \vec{m}$ (respectively $\vec{m}_0
  \xrightarrow{}_\Q^* \vec{m}'$ for some $\vec{m}'\ge \vec{m}$).
\end{definition}

Recently $\Q$-reachability and $\Q$-coverability were shown to be
decidable in polynomial time~\cite{FH15}.  In
Section~\ref{ssec:q-coverability}, we will discuss in detail the
approach from~\cite{FH15}. For now, observe that $\vec{m}
\xrightarrow{} \vec{m}'$ implies $\vec{m} \xrightarrow{}_\Q \vec{m}'$,
and hence $\vec{m} \xrightarrow{}^* \vec{m}'$ implies $\vec{m}
\xrightarrow{}_\Q^* \vec{m}'$. Consequently, $\Q$-coverability
provides an over-approximation of coverability: this fact is the
cornerstone of this paper.

\subsubsection{Upward Closed Sets.}\label{ssec:ucs}

A set $V \subseteq \N^P$ is \emph{upward closed} if for every $\vec{v}
\in V$ and $\vec{w} \in \N^P$, $\vec{v} \le \vec{w}$ implies $\vec{w}
\in V$. The \emph{upward closure} of a vector $\vec{v} \in \N^P$ is
the set $\uwclosure \vec{v} \defeq \{\vec{w}\in \N^P : \vec{v} \le
\vec{w}\}$. This definition can be lifted to sets $V \subseteq \N^P$
in the obvious way, \ie, $\uwclosure V \defeq \bigcup_{\vec{v} \in V}
\uwclosure \vec{v}$. Due to $\N^P$ being well-quasi-ordered by $\le$,
any upward-closed set $V$ contains a finite set $F \subseteq V$ such
that $V = \uwclosure F$. Such an $F$ is called a \emph{basis} of $V$
and allows for a finite representation of an upward-closed set. In
particular, it can be shown that $V$ contains a unique \emph{minimal
  basis} $B \subseteq V$ that is minimal with respect to inclusion for
all bases $F \subseteq V$. We denote $\Min(F)$ this minimal basis
obtained by deleting vectors $\vec{v}\in F$ such that there exists
$\vec{w}\in F$ with $\vec{w}<\vec{v}$ (when $F$ is finite).


\section{Deciding Coverability and $\Q$-Reachability}\label{sec:existing-alg}

We now introduce and discuss existing algorithms for solving
coverability and $\Q$-reachability which form the basis of our
approach. The main reason for doing so is that it allows us to
smoothly introduce some additional notations and concepts that we
require in the next section. For the remainder of this section, we fix
some Petri net system $\pns = (\net, \vec{m}_0)$ with $\net = (P, T,
\Pre, \Post)$, and some marking $\vec{m}$ to be covered or
$\Q$-reached.

\subsection{The Backward Coverability Algorithm}

The standard backward coverability algorithm,
Algorithm~\ref{alg:backward}, is a simple to state algorithm.
\begin{itemize}
 \item It iteratively constructs minimal bases $M$, where in the
   $k$-th iteration, $M$ is the minimal basis of the (upward closed)
   set of markings that can cover $\vec{m}$ after a firing sequence of
   length at most $k$. If $\vec{m}_0 \in \uwclosure M$, the algorithm
   returns true, i.e., that $\vec{m}$ is coverable. Otherwise, in
   order to update $M$, for all $\vec{m}' \in M$ and $t \in T$ it
   computes $\vec{m}'_t(p) \defeq \max \{\Pre(p, t),\ \vec{m}'(p) -
   \mat C(p, t)\}$.  The singleton $\{\vec{m}'_t\}$ is the minimal
   basis of the set of vectors that can cover $\vec{m}'$ after firing
   $t$.
\item Thus defining $pb(M)$ as $\pb{M} \defeq \bigcup_{\vec{m'} \in M,
  t \in T} \{\vec{m}_t'\}$, $M\cup pb(M)$ is a (not necessarily minimal)
  basis of the upward closed set of markings that can cover $\vec{m}$
  after a firing sequence of length at most $k+1$. This basis can be
  then minimized in every iteration.
\end{itemize}

%

\begin{algorithm}[t]
  \begin{algorithmic}[1]
    \REQUIRE PNS $\pns=(\net, \vec{m}_0)$  and a marking $\vec{m}\in \N^P$
    \STATE $M := \{ \vec{m} \}$;
    \WHILE{$\vec{m}_0 \not \in \uwclosure M$}
    \STATE $B := \pb{M} \setminus \uwclosure M$;
    \IF {$B = \emptyset$}
    \RETURN $\mathit{false}$;
    \ELSE
    \STATE $M := \Min(M \cup B)$;
    \ENDIF
    \ENDWHILE
    \RETURN $\mathit{true}$;
  \end{algorithmic}
  \caption{Backward Coverability}
  \label{alg:backward}
\end{algorithm}

The termination of the algorithm is guaranteed due to $\N^P$ being
well-quasi-ordered, which entails that $M$ must stabilize and return
false in this case. It can be shown that Algorithm~\ref{alg:backward}
runs in \ComplexityFont{2}-\EXP~\cite{BG11}. The key point to the
(empirical) performance of the algorithm is the size of the set $M$
during its computation: the smaller, the better. Even though one can
establish a doubly-exponential lower bound on the cardinality of $M$
during the execution of the algorithm, in general not every element in
$M$ is coverable, even when $\vec{m}$ is coverable.

\subsection{The $\Q$-Reachability Algorithm}\label{ssec:q-coverability}

We now present the fundamental concepts of the polynomial-time
$\Q$-reachability algorithm of Fraca and Haddad~\cite{FH15}. The key
insight underlying their algorithm is that $\Q$-reachability can be
characterized in terms of three simple criteria. The algorithm relies
on the notions of \emph{firing set} and \emph{maximal firing set},
denoted $\fs(\net, \vec{m})$ and $\maxfs{\net}{\vec{m}}$, and defined
as follows:
\begin{align*}
  \fs(\net,\vec{m})\ & \defeq\ \{\eval{\parikh(\sigma)} : \sigma \in
  (\Qpos \times T)^*,\ \text{there is } \vec{m}' \in \Qpos^P \text{ s.t.\ } 
  \vec{m} \xrightarrow{\sigma}_\Q \vec{m}'\} \\
  \maxfs{\net}{\vec{m}}\ & \defeq\ \bigcup_{T'\in \fs(\net, \vec{m})} T'.
\end{align*}
Thus, $\fs(\net, \vec{m})$ is the set of supports of firing sequences
starting in $\vec{m}$. Even though $\fs(\net, \vec{m})$ can be of size
exponential with respect to $\abs{T}$, deciding $T' \in \fs(\net,
\vec{m})$ for some $T' \subseteq T$ can be done in polynomial time,
and $\maxfs{\net}{\vec{m}}$ is also computable in polynomial
time~\cite{FH15}. The following proposition characterizes the set of
$\Q$-reachable markings.
\begin{proposition}[{\cite[Thm.\ 20]{FH15}}]\label{prop:q-reachability}
  A marking $\vec{m}$ is $\Q$-reachable in $\pns = (\net, \vec{m}_0)$
  if and only if there exists $\vec{x} \in \Qpos^T$ such that
  \begin{enumerate}[(i)]
  \item $\vec{m} = \vec{m}_0 + \mat{C} \cdot \vec{x}$
  \item $\eval{\vec{x}} \in \fs(\net, \vec{m}_0)$
  \item $\eval{\vec{x}} \in \fs(\net^{-1}, \vec{m})$
  \end{enumerate}
\end{proposition}

In this characterization, $\vec{x}$ is supposed to be the Parikh image
of a firing sequence.  The first item expresses the state equation of
$\pns$ with respect to $\vec{m}_0$, $\vec{m}$ and $\vec{x}$.  The two
subsequent items express that the support of the solution of the state
equation has to lie in the firing sets of $\pns$ and its reverse. As
such, the characterization in Proposition~\ref{prop:q-reachability}
yields an \NP{} algorithm. By employing a greatest fixed point
computation, Algorithm~\ref{alg:q-reachability}, which is a decision
variant of the algorithm presented in~\cite{FH15}, turns those
criteria into a polynomial-time algorithm (see~\cite{FH15} for a proof
of its correctness).
\begin{algorithm}[t]
  \begin{algorithmic}[1]
    \REQUIRE PNS $\pns=(\net, \vec{m}_0)$ 
    with $\net=(P, T, \Pre, \Post)$ and a marking $\vec{m}$
    \IF{$\vec{m}=\vec{m}_0$}
    \RETURN $\mathit{true}$;
    \ENDIF
    \STATE $T' := T$;
    \WHILE{$T'\neq \emptyset$} 
    \STATE $S := \emptyset$;
    \FORALL{$t \in T'$}
      \STATE ${\vec{x}} := \mathit{solve}(\mat{C}_{P \times T'} \cdot \vec{x} =
      \vec{m} - \vec{m}_0 \wedge \vec{x}(t) > 0 \wedge \vec{x} \in \Qpos^{T'})$;
      \IF{$\vec{x} \neq \mathit{undef}$}
      \STATE $S := S \cup \eval{\vec{x}}$;
      \ENDIF
    \ENDFOR
    \IF{$S=\emptyset$}
    \RETURN $\mathit{false}$;
    \ENDIF
    \STATE $T' := \maxfs{\net_{S}}{\vec{m}_0[\prepost{S}]}) \cap
    \maxfs{\net_{S}^{-1}}{\vec{m}[\prepost{S}]})$
    \IF{$T' = S$}
    \RETURN $\mathit{true}$;
    \ENDIF
    \ENDWHILE
    \RETURN $\mathit{false}$;
  \end{algorithmic}
  \caption{$\Q$-reachability~\cite{FH15}}
  \label{alg:q-reachability}
\end{algorithm}
In order to use Algorithm~\ref{alg:q-reachability} for deciding
coverability, it is sufficient, for each place $p$, to add a
transition to $\net$ that can at any time non-deterministically
decrease $p$ by one token. Denote the resulting Petri net system by
$\pns'$, it can easily checked that $\vec{m}$ is $\Q$-coverable in
$\pns$ if and only if $\vec{m}$ is $\Q$-reachable in $\pns'$.


\section{Backward Coverability Modulo $\Q$-Reachability}\label{sec:approach}
We now present our decision algorithm for the Petri net coverability
problem.

\subsection{Encoding $\Q$-Reachability into Existential FO($\Qpos,+,>$)}\label{ssec:encoding}
Throughout this section, when used in formulas, $\vec{w}$ and
$\vec{x}$ are vectors of first-order variables indexed by $P$
representing markings, and $\vec{y}$ is a vector of first-order
variables indexed by $T$ representing the $\Q$-Parikh image of a
transition sequence.

Condition~(i) of Proposition~\ref{prop:q-reachability}, which
expresses the state equation, is readily expressed as a system of
linear equations and thus directly corresponds to a formula
$\Phi(\vec{w}, \vec{x},\vec{y})$ which holds whenever a marking
$\vec{x}$ is reached starting at marking $\vec{w}$ by firing every
transition $\vec{y}(t)$ times (without any consideration whether such
a firing sequence would actually be admissible):
\[
\Phi_{\mathit{eqn}}^{\net}(\vec{w}, \vec{x}, \vec{y}) \defeq \vec{x} = \mat{C} \cdot 
\vec{y} + \vec{w}.
\]
Next, we show how to encode Conditions~(ii) and~(iii) into suitable
formulas. To this end, we require an effective characterization of
membership in the firing set $\fs(\net, \vec{w})$ defined in
Section~\ref{ssec:q-coverability}. The following characterization can
be derived from~\cite[Cor.~19]{FH15}. First, we define a monotonic increasing
function $\incfs_{\net, \vec{w}} : 2^T \to 2^T$ as follows:
\[
\incfs_{\net, \vec{w}}(S) \defeq
  S \cup \left\{ t \in T(\net) : \pre{t} \subseteq \eval{\vec{w}} \cup 
  \{ \post{s} : s\in S\} \right\}.
\]
From~\cite[Cor.~19]{FH15}, it follows that $T'\in \fs(\net,\vec{w})$
if and only if $T' = \lfp(\incfs_{\net_{T'}, \vec{w}})$, where $\lfp$
is the least fixed point operator\footnote{In~\cite[Cor.~19]{FH15}, an
  algorithm is presented that basically computes
  $\lfp(\incfs_{\net_{T'}, \vec{w}})$.}, i.e.,
\[T' = \incfs_{\net_{T'},
  \vec{w}}(\cdots (\incfs_{\net_{T'}, \vec{w}}(\emptyset))\cdots
).
\]
Clearly, the least fixed
point is reached after at most $\abs{T'}$  iterations.

In order to decide whether $\eval{\vec{y}} \in \fs(\net,\vec{w})$, we
simulate this fixed-point computation in an existential
FO($\Qpos,+,>$)-formula $\Phi_{\fs}^\net(\vec{w},\vec{y})$. Our
approach is inspired by a technique of Verma, Seidl and Schwentick
that was used to show that the reachability relation for
communication-free Petri nets is definable by an existential
Presburger arithmetic formula of linear size~\cite{VSS05}. The basic
idea is to introduce additional first-order variables $\vec{z}$
indexed by $P\cup T$ that, given a firing set, capture the relative
order in which transitions of this set are fired and the order in
which their input places are marked.
This order corresponds to the computation of
$\lfp(\incfs_{\net_{\eval{\vec{y}}},\vec{w}})$ and is encoded via a
numerical value $\vec{z}(t)$ (respectively $\vec{z}(p)$), representing
an index that must be strictly greater than zero for a transition
(respectively an input place of a transition) of this set. In
addition, input places have to be marked before the firing of a
transition:
\[
\Phi_{\mathit{dt}}^\net(\vec{y}, \vec{z}) \defeq 
\bigwedge_{t\in T}\left( \vec{y}(t) > 0 \rightarrow \bigwedge_{p\in {\pre{t}}} 
0<\vec{z}(p) \leq \vec{z}(t)\right).
\]

Moreover, a place is either initially marked or after the firing of a
transition of the firing set. So:
%
\begin{multline*}
  \Phi_{\mathit{mk}}^{\net} (\vec{w},\vec{y}, \vec{z}) \defeq
  \bigwedge_{p\in P}  \left(
    \vec{z}(p) > 0 \rightarrow \left(\vec{w}(p) > 0 \vee
    \bigvee_{t\in \pre{p}} \vec{y}(t) > 0 \wedge 
    \vec{z}(t)<\vec{z}(p) \right)  \right).  
\end{multline*}
We can now take the conjunction of the formulas above in order to
obtain a logical characterization of $\fs(\net,\vec{w})$:
\[
\Phi_\fs^\net(\vec{w},\vec{y}) \defeq \exists \vec{z} : 
\Phi_{\mathit{dt}}^\net(\vec{y},\vec{z}) \wedge \Phi_{\mathit{mk}}^\net(\vec{w},
\vec{y},\vec{z}).
\]

Having logically characterized all conditions of
Proposition~\ref{prop:q-reachability}, we can define the global
$\Q$-reachability relation for a Petri net system $\pns = (\net,
\vec{w})$ as follows:
\[
\Phi_{\pns}(\vec{w}, \vec{x}) \defeq 
\exists \vec{y} : \Phi_{\mathit{eqn}}^\net(\vec{w},\vec{x},\vec{y}) \wedge 
\Phi_\fs^\net(\vec{w}, \vec{y}) \wedge \Phi_\fs^{\net^{-1}}(\vec{x}, \vec{y}).
\]
In summary, we have thus proved the following result in this section.
\begin{proposition}\label{prop:q-fo-reachability}
  Let $\pns=(\net, \vec{m}_0)$ be a Petri net system and $\vec{m}$ be
  a marking. There exists an existential
  $\text{FO}(\Qpos,+,>)$-formula $\Phi_\pns(\vec{w}, \vec{x})$ 
  computable in linear time  such
  that $\vec{m}$ is $\Q$-reachable in $\pns$ if and only if
  $\Phi_\pns(\vec{m}_0,\vec{m})$ is valid.
\end{proposition}
Checking satisfiability of $\Phi_\pns$ is in \NP, see
e.g.~\cite{Son85}. It is a valid question to ask why one would prefer
an \NP-algorithm over a polynomial-time one. We address this question
in the next section. For now, note that in order to obtain an even
more accurate over-approximation, we can additionally restrict
$\vec{y}$ to be interpreted in the natural numbers while retaining
membership of satisfiability in \NP, due to the following variant of
Proposition~\ref{prop:q-reachability}: If a marking is reachable in
$\pns$ then there exists some $\vec{y}\in \N^T$ such that
Conditions~(i), (ii) and~(iii) of
Proposition~\ref{prop:q-reachability} hold.
\begin{remark}
  Proposition~\ref{prop:q-fo-reachability} additionally allows us to
  improve the best known upper bound for the \emph{inclusion problem}
  of continuous Petri nets, which is \EXP~\cite{FH15}. Given two Petri
  net systems $\pns=(\net, \vec{m}_0)$ and $\pns'=(\net', \vec{m}_0')$
  over the same set of places, this problem asks whether the set of
  reachable markings of $\pns$ is included in $\pns'$, i.e., whether
  $\forall \vec{m}.\Phi_{\pns} (\vec{m}_0,\vec{m}) \rightarrow
  \Phi_{\pns'}(\vec{m}_0',\vec{m})$ is valid. The latter is a
  $\Pi_2$-sentence of $\text{FO}(\Qpos,+,>)$ and decidable in
  $\ComplexityFont{\Pi_2^\P}$~\cite{Son85}.  Hence, inclusion between
  continuous Petri nets is in $\ComplexityFont{\Pi_2^\P}$.
\end{remark}


\subsection{The Coverability Decision Procedure}\label{ssec:backward-coverability}
We now present Algorithm~\ref{alg:backward-modulo-q} for deciding
coverability. This algorithm is an extension of the classical backward
reachability algorithm that incorporates $\Q$-reachability checks
during its execution in order to keep the set of minimal basis
elements small.

\begin{algorithm}[t]
  \begin{algorithmic}[1]
    \REQUIRE PNS $\pns=(\net, \vec{m}_0)$ 
    and a marking $\vec{m}\in \N^P$
    \STATE $M := \{ \vec{m} \}$; $\Phi(\vec{x}) := \exists \vec{y}: 
    \Phi_\pns(\vec{m}_0, \vec{y}) \wedge \vec{y} \ge \vec{x}$;
    \IF{\NOT $\Q$-$\mathit{coverable}(\pns, \vec{\vec{m}})$}
    \RETURN $\mathit{false}$
    \ENDIF
    \WHILE{$\vec{m}_0 \not \in \uwclosure M$}
    \STATE $B := \pb{M} \setminus \uwclosure M$;
    \STATE $D := \{ \vec{v} \in B : \unsat(\Phi(\vec{v})) \}$;
    \STATE $B := B \setminus D$;
    \IF {$B = \emptyset$}
    \RETURN $\mathit{false}$;
    \ELSE
    \STATE $M := \Min(M \cup B)$;
    \STATE $\Phi(\vec{x}) := \Phi(\vec{x}) \wedge \bigwedge_{\vec{v}\in D}
    \vec{x} \not \ge \vec{v}$;
    \ENDIF
    \ENDWHILE
    \RETURN $\mathit{true}$;
  \end{algorithmic}
  \caption{Backward Coverability Modulo $\Q$-Reachability}
  \label{alg:backward-modulo-q}
\end{algorithm}

First, on Line~1 we derive an open formula $\Phi(\vec{x})$ from
$\Phi_\pns$ such that $\Phi(\vec{x})$ holds if and only if $\vec{x}$
is $\Q$-coverable in $\pns$. Then, on Line~2, the algorithm checks
whether the marking $\vec{m}$ is $\Q$-coverable using the
polynomial-time algorithm from~\cite{FH15} and returns that $\vec{m}$
is not coverable if this is not the case. Otherwise, the algorithm
enters a loop which iteratively computes a basis $M$ of the backward
coverability set starting at $\vec{m}$ whose elements are in addition
$\Q$-coverable in $\pns$. To this end, on Line~5 the algorithm
computes a set $B$ of new basis elements obtained from one application
of $\mathit{pb}$, and on Line~7 it removes from $B$ the set $D$ which
contains all elements of $B$ which are not $\Q$-coverable. If as a
result $B$ is empty the algorithm concludes that $\vec{m}$ is not
coverable in $\pns$. Otherwise, on Line~11 it adds the elements of $B$
to $M$. Finally, Line~12 makes sure that in future iterations of the
loop the underlying SMT solver can immediately discard elements that
lie in $\uwclosure D$. The latter is technically not necessary, but it
provides some guidance to the SMT solver. \iftechrep{The proof of the
  following proposition is deferred to
  Appendix~\ref{app:approach-correctness}.}{The proof of the following
  proposition can be found in~\cite{BFHH15}.}

\begin{proposition}\label{prop:approach-correctness}
  Let $\pns = (\net, \vec{m}_0)$ be a PNS and $\vec{m}$ be a marking.
  Then $\vec{m}$ is coverable in $\pns$ if and only if
  Algorithm~\ref{alg:backward-modulo-q} returns $\mathit{true}$.
\end{proposition}

\begin{remark}
  In our actual implementation, we use a slight variation of
  Algorithm~\ref{alg:backward-modulo-q} in which the instruction $M :=
  \Min(M \cup B)$ in Line~11 is replaced by $M := \Min(M \cup
  \minbottle_{c,k} B)$. Here, $c,k\in \N$ are parameters to the
  algorithm, and $\minbottle_{c,k} B$ is the set of the $c+\abs{B}/k$
  elements of $B$ with the smallest sum-norm. In this way, the
  empirically chosen parameters $c$ and $k$ create a bottleneck that
  gives priority to elements with small sum-norms, as they are more
  likely to allow for discarding elements with larger sum-norms in
  future iterations.
  
  This variation of Algorithm~\ref{alg:backward-modulo-q} has the same
  correctness properties as the original one: It can be shown that
  using $\minbottle_{c,k} B$ instead of $B$ in Line~11 computes the
  same set $\uwclosure{M}$ at the expense of delaying its
  stabilization. \label{remark:bottleneck}
\end{remark}


Before we conclude this section, let us come back to the question why
in our approach we choose using $\Phi_\pns$ (whose satisfiability is
in \NP) over Algorithm~\ref{alg:q-reachability} which runs in
polynomial time. In Algorithm~\ref{alg:backward-modulo-q}, we invoke
Algorithm~\ref{alg:q-reachability} only once in Line~2 in order to
check if $\pns$ is not $\Q$-coverable, and thereafter only employ
$\Phi_\pns$ which gets incrementally updated during each iteration of
the loop. The reason is that in practice as observed in our
experimental evaluation below, Algorithm~\ref{alg:q-reachability}
turns out to be often faster for a \emph{single} $\Q$-coverability
query. Otherwise, as soon $\Phi_\pns$ has been checked for
satisfiability once, future satisfiability queries are significantly
faster than Algorithm~\ref{alg:q-reachability}, which is a desirable
behavior inside a backward coverability framework. Moreover we can
constraint solutions to be in $\N$ instead of $\Q$, leading to a more
precise over approximation.


\subsection{Relationship to the CEGAR-approach of Esparza et al.}\label{ssec:esparza}
In~\cite{ELMMN14}, Esparza~\emph{et al.} presented a semi-decision
procedure for coverability that is based on~\cite{EM00} and employs
the Petri net state equation and traps inside a CEGAR-framework. A
\emph{trap in $\net$} is a non-empty subset of places $Q\subseteq P$
such that $\post{Q} \subseteq \pre{Q}$, and $Q\subseteq P$ is a
\emph{siphon in $\net$} whenever $\pre{Q} \subseteq \post{Q}$. Given a
marking $\vec{m}$, a trap (respectively siphon) is \emph{marked in
  $\vec{m}$} if $\sum_{p\in Q} \vec{m}(p)>0$. An important property of
traps is that if a trap is marked in $\vec{m}$, it will remain marked
after any firing sequence starting in $\vec{m}$. Conversely, when a
siphon is unmarked in $\vec{m}$ it remains so after any firing
sequence starting in $\vec{m}$. By definition, $Q$ is a trap in $\net$
if and only if $Q$ is a siphon in $\net^{-1}$. The coverability
criteria that~\cite{ELMMN14} builds upon are derived from~\cite{EM00}
and can be summarized as follows.

\begin{proposition}[\cite{ELMMN14}]\label{prop:esparza-conditions}
  If $\vec{m}$ is $\Q$-reachable (respectively reachable) in
  $(\net,\vec{m}_0)$ then there exists $\vec{x} \in \Qpos^T$
  (respectively $\vec{x} \in \N^T$) such that:
  \begin{enumerate}[(i)]
  \item $\vec{m} = \vec{m}_0 + \mat{C} \cdot \vec{x}$
  \item for all traps $Q\subseteq P$, if $Q$ is marked in $\vec{m}_0$
    then $Q$ is marked in $\vec{m}$
  \end{enumerate}
\end{proposition}
As in our approach, in~\cite{ELMMN14} those criteria are checked using
an SMT-solver. The for-all quantifier is replaced in~\cite{ELMMN14} by
incrementally enumerating all traps in a CEGAR-style fashion. It is
shown in~\cite[Prop.\ 18]{FH15} that Condition~(iii) of
Proposition~\ref{prop:q-reachability} is equivalent to requiring that
$\net_{\eval{\vec{x}}}^{-1}$ has no unmarked siphon in $\vec{m}$,
which appears to be similar to Condition~(ii) of
Proposition~\ref{prop:esparza-conditions}. In fact, we show the
following.
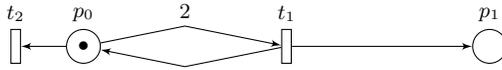
\begin{figure}[t]
\begin{center}
\begin{tikzpicture}[xscale=1,yscale=1,scale=0.9, every node/.style={scale=0.9}]

\path (0,0) node[draw,circle,inner sep=2pt,minimum height=0.5cm] (s0) [label=above:$p_0$] {$\bullet$};
\path (3,0) node[draw,rectangle,inner sep=2pt,minimum height=0.5cm] (s1)
[label=above:$t_1$] {};
\path (6,0) node[draw,circle,inner sep=2pt,minimum height=0.5cm] (s2)[label=above:$p_1$] {};
\path (-1,0) node[draw,rectangle,inner sep=2pt,minimum height=0.5cm] (s3)
[label=above:$t_2$] {};


\draw[arrows=-latex'] (s0) -- (1.5,0.3) -- (s1) node[pos=0,above] {$2$};
\draw[arrows=-latex'] (s1) -- (1.5,-0.3) -- (s0) node[pos=.5,above] {};
\draw[arrows=-latex'] (s1) -- (s2) node[pos=.5,above] {};
\draw[arrows=-latex'] (s0) -- (s3) node[pos=.5,above] {};

%
%
\end{tikzpicture}
\end{center}
  
  \caption{A Petri net that cannot mark $p_1$.}
  \label{fig:strictness}
\end{figure}
\begin{proposition}\label{lem:esparza-comparison}
  Conditions~(i) and~(iii) of Proposition~\ref{prop:q-reachability}
  strictly imply
  Conditions~(i) and~(ii) of Proposition~\ref{prop:esparza-conditions}
  (when interpreted over $\Qpos$).
\end{proposition}
\begin{proof}
  We only show strictness, \iftechrep{the remainder of the proof is
    deferred to Appendix~\ref{app:esparza-comparison}.}{the full proof
    can be found in~\cite{BFHH15}.} To this end, consider the Petri
  net $(\net,\vec{m}_0)$ depicted in Figure~\ref{fig:strictness} with
  $\vec{m}=(0,1)$.  Clearly $\vec{m}$ is not reachable.  There is a
  single solution to the state equation $\vec{x}=(1,0)$.  There is a
  single trap $\{p_1\}$ which is unmarked in $\vec{m}_0$. So the
  conditions of Proposition~\ref{prop:esparza-conditions} hold, and
  hence the algorithm of~\cite{ELMMN14} does not decide this net
  safe. On the contrary in $\net_{\eval{\vec{x}}}^{-1}$, the reverse
  net without $t_2$, $\{p_0\}$ is a siphon that is unmarked in
  $\vec{m}$. So Condition~(iii) of
  Proposition~\ref{prop:q-reachability} does not hold. \qed
\end{proof}
This proposition shows that the single formula stated in
Proposition~\ref{prop:q-fo-reachability} strictly subsumes the
approach from~\cite{ELMMN14}. Moreover, it provides a theoretical
justification for why the approach of~\cite{ELMMN14} performs so well
in practice: the conditions are a strict subset of the conditions
developed for $\Q$-reachability in~\cite{FH15}.



\section{Experimental Evaluation}\label{sec:benchmarks}

We evaluate the backward coverability modulo $\Q$-reachability
algorithm on standard benchmarks from the literature with two goals in
mind. First, we demonstrate that our approach is competitive with
existing approaches. In particular, we prove significantly more safe
instances of our benchmarks safe in less time when compared to any
other approach. Overall our algorithm decides 142 out of 176
instances, the best competitor decides 122 instances. Second, we
demonstrate that $\Q$-coverability is a powerful pruning criterion by
analyzing the relative number of minimal bases elements that get
discarded during the execution of
Algorithm~\ref{alg:backward-modulo-q}.

We implemented Algorithm~\ref{alg:backward-modulo-q} in a tool called
\tool{} in the programming language
\textsc{Python}.\footnote{\scriptsize\tool{} is available at
  \url{http://www-etud.iro.umontreal.ca/~blondimi/qcover/}.} The
underlying SMT-solver is \textsc{z3}~\cite{z3}. For the
$\minbottle_{c,k}$ heuristic mentioned in
Remark~\ref{remark:bottleneck}, we empirically chose $c = 10$ and
$k=5$. We observed that any sane choice of $c$ and $k$ leads to an
overall speed-up, though different values lead to different (even
increasing) running times on individual instances. \tool{} takes as
input coverability instances in the \mist{} file
format.\footnote{\scriptsize\url{https://github.com/pierreganty/mist/wiki\#input-format-of-mist}}
The basis of our evaluation is the benchmark suite that was used in
order to evaluate the tool \petrinizer, see~\cite{ELMMN14} and the
references therein. This suite consists of five benchmark categories:
\texttt{mist}, consisting of 27 instances from the \mist{} toolkit;
\texttt{bfc}, consisting of 46 instances used for evaluating BFC;
\texttt{medical} and \texttt{bug\_tracking}, consisting of 12 and 41
instances derived from the provenance analysis of messages of a
medical and a bug-tracking system, respectively; and \texttt{soter},
consisting of 50 instances of verification conditions derived from
Erlang programs~\cite{DOsualdoKO13}.

We compare \tool{} with the following tools:
\petrinizer{}~\cite{ELMMN14}, \mist{}~\cite{Gan02} and
\bfc{}~\cite{KKW14} in their latest versions available at the time of
writing of this paper. \mist{} implements a number of algorithms, we
use the backward algorithm that uses places invariant
pruning~\cite{GMDKRB07}.\footnote{\scriptsize
  \url{https://github.com/pierreganty/mist/wiki\#coverability-checkers-included-in-mist}}
All benchmarks were performed on a single computer equipped with four
Intel\textregistered{} Core\texttrademark{} 2.00 GHz
%
%
CPUs, 8 GB of memory and Ubuntu Linux 14.04 (64 bits). The execution
time of the tools was limited to 2000 seconds (\ie\ 33 minutes and 20
seconds) per benchmark instance. The running time of every tool on an
instance was determined using the sum of the user and sys time
reported by the Linux tool \texttt{time}.

\begin{figure}[t]
 


\begin{center}
\resizebox{\columnwidth}{!}{%
\begin{tabular}{cc}
\begin{tabular}{|c|c|c|c|c||c|}
 \hline Suite & \tool & \petrinizer & \mist & \bfc & Total \\
\hhline{|=|=|=|=|=|=|} \texttt{mist} & \textbf{23} & 20 & 22 & 20 & 23 \\
\hline \texttt{medical} & \textbf{11} & 4 & \textbf{11} & 3 & 12 \\
\hline \texttt{bfc} & \textbf{2} & \textbf{2} & \textbf{2} & \textbf{2} & 2 \\
\hline \texttt{bug\_tracking} & \textbf{32} & \textbf{32} & 0 & 19 & 40\\
\hline \texttt{soter} & \textbf{37} & \textbf{37} & 0 & 19 & 38 \\
\hhline{|=|=|=|=|=|=|} Total & \textbf{105} & 95 & 35 & 63 & 115 \\
\hline
\end{tabular} \hspace{0.20cm} &
\begin{tabular}{|c|c|c|c|c||c|}
 \hline Suite & \tool & \petrinizer & \mist & \bfc & Total \\
\hhline{|=|=|=|=|=|=|} \texttt{mist} & 3 & --- & \textbf{4} & \textbf{4} & 4\\
\hline \texttt{medical} & --- & --- & --- & --- & 0 \\
\hline \texttt{bfc} & 26 & --- & 29 & \textbf{42} & 44 \\
\hline \texttt{bug\_tracking} & 0 & --- & 0 & \textbf{1} & 1 \\
\hline \texttt{soter} & 8 & --- & 6 & \textbf{12} & 12 \\
\hhline{|=|=|=|=|=|=|} Total & 37 & 0 & 39 & \textbf{59} & 61 \\
\hline
\end{tabular} \\[1.5cm]
\multicolumn{2}{c}{
\begin{tabular}{|c|c|c|c|c||c|}
 \hline Suite & \tool & \petrinizer & \mist & \bfc & Total \\
\hhline{|=|=|=|=|=|=|} \texttt{mist} & \textbf{26} & 20 & \textbf{26} & 24 & 27 \\
\hline \texttt{medical} & \textbf{11} & 4 & \textbf{11} & 3 & 12 \\
\hline \texttt{bfc} & 28 & 2 & 31 & \textbf{44} & 46 \\
\hline \texttt{bug\_tracking} & \textbf{32} & \textbf{32} & 0 & 20 & 41 \\
\hline \texttt{soter} & \textbf{45} & 37 & 6 & 31 & 50 \\
\hhline{|=|=|=|=|=|=|} Total & \textbf{142} & 95 & 74 & 122 & 176 \\
\hline
\end{tabular}
}
\end{tabular}
}
\end{center}

  \caption{Number of safe instances (top-left), unsafe instances (top-right)
    and total instances (bottom) decided by every tool. Bold numbers
    indicate the tool(s) which decide(s) the largest number of
    instances in the respective category.}
  \label{fig:results-table}
\end{figure}

Figure~\ref{fig:results-table} contains three tables which display the
number of safe instances shown safe, unsafe instances shown unsafe,
and the total number of instances of our benchmark suite decided by
each individual tool. As expected, our algorithm outperforms all
competitors on safe instances, since in this case a proof of safety
(\ie\ non-coverability) effectively requires the computation of the
whole backward coverability set, and this is where pruning via
$\Q$-coverability becomes most beneficial. On the other hand, \tool{}
remains competitive on unsafe instances, though a tool such as
\textsc{BFC} handles those instances better since its heuristics are
more suited for proving unsafety (\ie\ coverability). Nevertheless,
\tool{} is the overall winner when comparing the number of safe and
unsafe instances decided, being far ahead at the top of the
leader-board deciding 142 out of 176 instances.

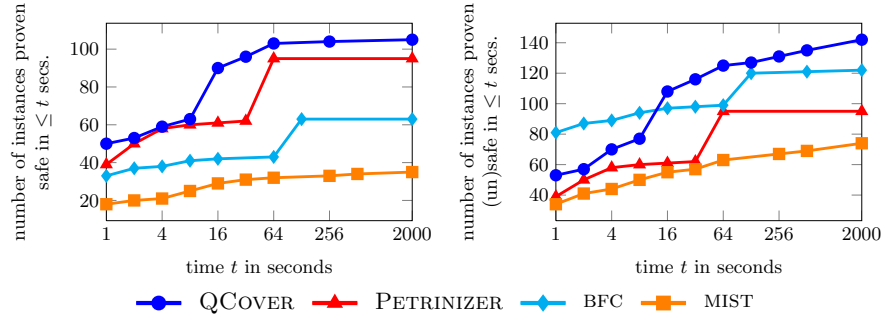
\begin{figure}
\begin{center}
\begin{tikzpicture}[scale=0.9, every node/.style={scale=0.9}]
  \begin{semilogxaxis}[
    xmin=1,
    xmax=2000,
    domain=1:2000,
    xtick={1, 4, 16, 64, 256, 2000},
    xticklabels={1, 4, 16, 64, 256, 2000},
    width=0.5\textwidth,
    height=4.5cm,
    xlabel=time $t$ in seconds,
    ylabel style={align=center},
    ylabel=number of instances proven \\ safe in $\leq t$ secs.,
    legend style={
      at={(0,0.65)},
      anchor=south west},
    reverse legend
    ]

    \addplot[very thick, color=red, mark=triangle*] coordinates {
      (1, 39)
      (2, 50)
      (4, 58)
      (8, 60)
      (16, 61)
      (32, 62)
      (64, 95)
      (2000, 95)
    };

    \addplot[very thick, color=cyan, mark=diamond*] coordinates {
      (1, 33)
      (2, 37)
      (4, 38)
      (8, 41)
      (16, 42)
      (64, 43)
      (128, 63)
      (2000, 63)
    };

    \addplot[very thick, color=orange, mark=square*] coordinates {
      (1, 18)
      (2, 20)
      (4, 21)
      (8, 25)
      (16, 29)
      (32, 31)
      (64, 32)
      (256, 33)
      (512, 34)
      (2000, 35)
    };

    \addplot[very thick, color=blue, mark=*] coordinates {
      (1, 50)
      (2, 53)
      (4, 59)
      (8, 63)
      (16, 90)
      (32, 96)
      (64, 103)
      (256, 104)
      (2000, 105)
    };
  \end{semilogxaxis}
\end{tikzpicture}\;
\begin{tikzpicture}[scale=0.9, every node/.style={scale=0.9}]
  \begin{semilogxaxis}[
    xmin=1,
    xmax=2000,
    domain=1:2000,
    xtick={1, 4, 16, 64, 256, 2000},
    xticklabels={1, 4, 16, 64, 256, 2000},
    ytick={20, 40, 60, 80, 100, 120, 140},
    width=0.5\textwidth,
    height=4.5cm,
    xlabel=time $t$ in seconds,
    ylabel style={align=center},
    ylabel=number of instances proven \\ (un)safe in $\leq t$ secs.,
    legend style={draw=none, legend columns=-1},
    reverse legend
    ]

    \addplot[very thick, color=red, mark=triangle*] coordinates {
      (1, 39)
      (2, 50)
      (4, 58)
      (8, 60)
      (16, 61)
      (32, 62)
      (64, 95)
      (2000, 95)
    };

    \addplot[very thick, color=cyan, mark=diamond*] coordinates {
      (1, 81)
      (2, 87)
      (4, 89)
      (8, 94)
      (16, 97)
      (32, 98)
      (64, 99)
      (128, 120)
      (512, 121)
      (2000, 122)
    };

    \addplot[very thick, color=orange, mark=square*] coordinates {
      (1, 34)
      (2, 41)
      (4, 44)
      (8, 50)
      (16, 55)
      (32, 57)
      (64, 63)
      (256, 67)
      (512, 69)
      (2000, 74)
    };

    \addplot[very thick, color=blue, mark=*] coordinates {
      (1, 53)
      (2, 57)
      (4, 70)
      (8, 77)
      (16, 108)
      (32, 116)
      (64, 125)
      (128, 127)
      (256, 131)
      (512, 135)
      (2000, 142)
    };
  \end{semilogxaxis}
\end{tikzpicture}
\begin{tikzpicture}
    \begin{customlegend}[legend columns=-1,legend style={draw=none,column sep=1ex},legend entries={\tool, \petrinizer, \bfc, \mist}]
    \addlegendimage{blue,fill=blue,mark=*, ultra thick}
    \addlegendimage{red,fill=red,mark=triangle*, ultra thick}
    \addlegendimage{cyan,fill=cyan,mark=diamond*, ultra thick}
    \addlegendimage{orange,fill=orange,mark=square*, ultra thick}
    \end{customlegend}
\end{tikzpicture}

\end{center}
  \caption{Cumulative number of instances proven safe (left) and total
    number of instances decided (right) within a fixed amount of
    time.}
  \label{fig:results-graph}
\end{figure}

\tool{} not only decides more instances, it often does so faster than
its competitors. Figure~\ref{fig:results-graph} contains two graphs
which show the cumulative number of instances proven safe and the
total number of instances decided on all suites by each tool within a
certain amount of time. When it comes to safety, \tool{} is always
ahead of all other tools. However, when looking at all instances
decided, \textsc{BFC} first has an advantage. We observed that this
advantage occurs on instances of comparably small size. As soon as
large instances come into play, \tool{} wins the race. Besides
different heuristics used, one reason for this might be the choice of
the implementation language (C for BFC vs.\ \textsc{Python} for
\tool{}). In particular, BFC can decide a non-negligible number of
instances in less than 10ms, which \tool{} never achieves.

\begin{figure}[t]
\begin{center}
\begin{tikzpicture}
  \begin{axis}[
      xlabel=percentage $x$,
      ylabel style={align=center},
      ylabel=number of iterations with \\ $\geq x$\% elements pruned,
      xmin=0, 
      xmax=100,
      ymin=0,
      height=4cm,
      width=0.45\textwidth
    ]

    \addplot[color=blue, fill=blue, opacity=0.75, smooth] coordinates {
      (0, 542)
      (5, 483)
      (10, 475)
      (15, 461)
      (20, 455)
      (25, 441)
      (30, 416)
      (35, 397)
      (40, 373)
      (45, 355)
      (50, 334)
      (55, 292)
      (60, 263)
      (65, 229)
      (70, 200)
      (75, 179)
      (80, 162)
      (85, 130)
      (90, 103)
      (95, 60)
      (100, 2)
      }
    |- (axis cs:0,0) -- cycle;
    \end{axis}
\end{tikzpicture}
\;
\begin{tikzpicture}
  \begin{axis}[
      xlabel=percentage $x$,
      ylabel style={align=center},
      ylabel=number of iterations with \\ $x$\% elements pruned,
      xmin=0, 
      xmax=100,
      ymin=0,
      height=4cm,
      width=0.45\textwidth,
      extra x ticks={55.68},
      extra x tick style={
        xticklabel pos=right,
        xticklabels={average},
        xmajorgrids=true}
    ]

    \addplot[ybar, bar width=5.15, color=blue, fill=blue, opacity=0.75]
    coordinates {
      (2.5, 59)
      (7.5, 8)
      (12.5, 14)
      (17.5, 6)
      (22.5, 14)
      (27.5, 25)
      (32.5, 19)
      (37.5, 24)
      (42.5, 18)
      (47.5, 21)
      (52.5, 42)
      (57.5, 29)
      (62.5, 34)
      (67.5, 29)
      (72.5, 21)
      (77.5, 17)
      (82.5, 32)
      (87.5, 27)
      (92.5, 43)
      (97.5, 60)
    };
    \end{axis}
\end{tikzpicture}
\end{center}
  \caption{Number of times a certain percentage of basis elements was
    removed due to $\Q$-coverability pruning.}
  \label{fig:q-coverability-pruning}
\end{figure}
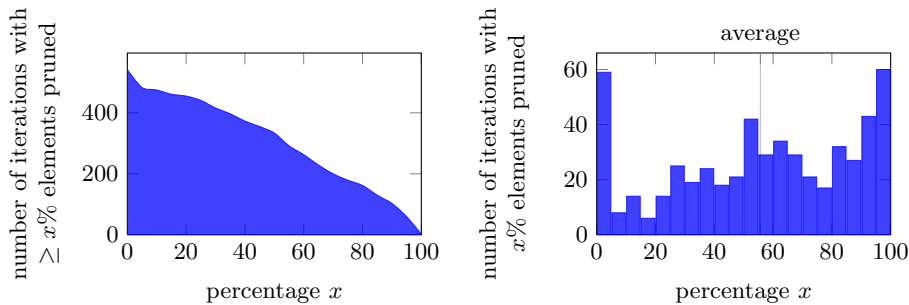

Finally, we consider the effectiveness of using $\Q$-coverability as a
pruning criterion. To this end, consider
Figure~\ref{fig:q-coverability-pruning} in which we plotted the number
of times a certain percentage of basis elements was removed due to not
being $\Q$-coverable. Impressively, in some cases more than 95\% of
the basis elements get discarded. Overall, on average we discard 56\%
of the basis elements, which substantiates the usefulness of using
$\Q$-coverability as a pruning criterion.

Before we conclude, let us mention that already 83 instances are
proven safe by only checking the state equation, and that additionally
checking for the criteria~(ii) and~(iii) of
Proposition~\ref{prop:q-reachability} increases this number to 101
instances. If we use Algorithm~\ref{alg:q-reachability} instead of our
FO($\Qpos,+,>$) encoding then we can only decide 132 instances in
total.  Finally, in our experiments, interpreting variables over $\Q$
instead of $\N$ resulted in no measurable overall performance gain.

In summary, our experimental evaluation shows that the backward
coverability modulo $\Q$-reachability approach to the Petri net
coverability problem developed in this paper is highly efficient when
run on real-world instances, and superior to existing tools and
approaches when compared on standard benchmarks from the literature.


\section{Conclusion}\label{sec:discussion}

In this paper, we introduced backward coverability modulo
$\Q$-reachability, a novel approach to the Petri net coverability
problem that is based on using coverability in continuous Petri nets
as a pruning criterion inside a backward coverability framework. A key
ingredient for the practicality of this approach is an existential
FO($\Qpos,+,>$)-characterization of continuous reachability, which we
showed to strictly subsume a recently introduced coverability
semi-decision procedure~\cite{ELMMN14}. Finally, we demonstrated that
our approach significantly outperforms existing ones when compared on
standard benchmarks.

There are a number of possible avenues for future work. It seems
promising to combine the forward analysis approach based on
incrementally constructing a Karp-Miller tree that is used in
BFC~\cite{KKW14} with the $\Q$-coverability approach introduced in
this paper. In particular, recently developed minimization and
acceleration techniques for constructing Karp-Miller trees should
prove beneficial, see e.g.~\cite{GRB10,RS13,BH14}. Another way to
improve the empirical performance of our algorithm is to internally
use more efficient data structures such as sharing
trees~\cite{DRB04}. It seems within reach that a tool which combines
all of the aforementioned techniques and heuristics could decide all
of the benchmark instances we used in this paper within reasonable
resource restrictions.

\smallskip
\noindent
\textbf{Acknowledgments.} 
We would like to thank Vincent Antaki for an early implementation of
Algorithm~\ref{alg:q-reachability}. We would also like to thank Gilles
Geeraerts for his support with the MIST file format.

\newpage
\bibliographystyle{plain}
\bibliography{bibliography}

\iftechrep{
\newpage
\begin{appendix}
  \section{Proof of Propositon~\ref{prop:approach-correctness}}\label{app:approach-correctness}

\begin{proof}
Let $B_n$ and
  $M_n$ be respectively the value of $B$ and $M$ at Line 7 and Line 11
  in the $n$-th iteration of the while-loop. It is possible to show by
  induction on $n$ that
  \begin{align*}
    \uwclosure{M_n} &= \bigcup_{\vec{y} \in \uwclosure{\vec{m}}}
    \left\{\vec{x} \in \N^d : \text{$\Q$-$\mathit{coverable}$}(\pns,
    \vec{\vec{x}}) \land \exists \sigma \in T^* \text{ s.t. } \vec{x}
    \xrightarrow{\sigma} \vec{y} \land \abs{\sigma} \leq n\right\}\\
    B_n & \subseteq \uwclosure{M_n} \setminus \uwclosure{M_{n-1}}.
  \end{align*}
  Since $\vec{m}$ can be covered from $\vec{m}_0$ if and only if
  $\Q$-$\mathit{coverable}(\pns, \vec{\vec{m}_0})$ and $\vec{m}_0
  \xrightarrow{}^* \vec{m}$, it follows that Line 9 and 13 are correct
  and thus that the algorithm is correct. Regarding termination, it is
  well-known that, for well-quasi-ordered sets, any inclusion chain of
  upward closed sets stabilizes. Therefore, the chain $\uwclosure{M_1}
  \subseteq \uwclosure{M_2} \subseteq \ldots$ stabilizes to some
  $\uwclosure{M_n}$. Thus, after a finite number of iterations, either
  $B$ becomes empty at Line 8 or $\vec{m}_0 \in \uwclosure{M}$ at Line
  4. Hence, the algorithm always terminates.\qed
\end{proof}

\section{Proof of Proposition~\ref{lem:esparza-comparison}}\label{app:esparza-comparison}

\begin{proof}
  We write $\mathit{marked}(Q,\vec{m})$ if $Q$ is marked
  in $\vec{m}$, and $\mathit{unmarked}(Q,\vec{m})$ otherwise.  We establish
  the contrapositive:  assuming that for all $\vec{x}\in \Qpos^T$
  some condition of Proposition~\ref{prop:esparza-conditions} does not
  hold, we show that for all $\vec{x}\in \Qpos^T$ some Condition~(i)
  and~(iii) of Proposition~\ref{prop:q-reachability} does not hold. To
  this end, let $\vec{x}\in \Qpos^T$, if Condition~(i) of
  Proposition~\ref{prop:esparza-conditions} does not hold, i.e.,
  $\vec{m} \neq \vec{m}_0 + \mat{C} \cdot \vec{x}$, then Condition~(i)
  of Proposition~\ref{prop:q-reachability} does not hold either. Thus,
  assume that $\vec{m} = \vec{m}_0 + \mat{C} \cdot \vec{x}$. We then
  have the following:
  \begin{align*}
    & \text{there is a trap $Q$ in $\net$ s.t.\
      $\mathit{marked}(Q,\vec{m}_0)$ and $\mathit{unmarked}(Q,\vec{m})$}\\ 
    \implies & \text{there is a siphon $Q$ in $\net^{-1}$ s.t.\
      $\mathit{marked}(Q,\vec{m}_0)$ and $\mathit{unmarked}(Q,\vec{m})$}\\
   \overset{(\ast)}{\implies} & \text{there is a siphon $Q'$ in $\net^{-1}_{T'}$ 
     s.t.\ $\mathit{unmarked}(Q',\vec{m})$, where $T' = \eval{\vec{x}}$}\\
    \implies & \eval{\vec{x}}\notin \fs(\net^{-1}, \vec{m}).
  \end{align*}
  In order to show the implication~$(\ast)$, we first observe that
  $\vec{m}[Q]\neq \vec{m}_0[Q]$, since $Q$ is marked in $\vec{m}_0$
  and not marked in $\vec{m}$. Let $T'\defeq \eval{\vec{x}}$ and
  $P'\defeq \prepost{T'}$, we claim that $Q'\defeq Q \cap P'$ is a
  siphon in $\net_{T'}$. Let $t\in T'$ and suppose that $t\in
  \pre{Q'}$.  Since $Q'\subseteq Q$, we have $t\in \pre{Q}$ and hence
  $t\in \post{Q}$, which yields $t\in \post{Q'}$ as all the neighbour
  places of $t$ belong to $P'$.  It remains to show that
  $Q'\neq\emptyset$. To the contrary, assume that $Q' = \emptyset$,
  then $\vec{x}(t)=0$ for every $t\in \prepost{Q}$ and hence
  $\vec{m}[Q] = \vec{m}_0[Q]$, a contradiction. \qed
\end{proof}

\vfill

\end{appendix}
}{}

\end{document}
